\documentclass{article}

\usepackage{arxiv}

\usepackage{amsmath,amssymb,amsthm}
\theoremstyle{plain}
\newtheorem{theorem}{Theorem}[section]

\newtheorem{corollary}[theorem]{Corollary}

\theoremstyle{definition}
\newtheorem{definition}[theorem]{Definition}

\theoremstyle{remark}
\newtheorem{remark}[theorem]{Remark}

\usepackage[utf8]{inputenc} 
\usepackage[T1]{fontenc}    
\usepackage{hyperref}       
\usepackage{url}            
\usepackage{booktabs}       
\usepackage{amsfonts}       
\usepackage{nicefrac}       
\usepackage{microtype}      
\usepackage{cleveref}       
\usepackage{lipsum}         
\usepackage{graphicx}
\usepackage{natbib}
\usepackage{doi}

\usepackage{tikz}
\usetikzlibrary{arrows.meta, positioning, shapes.geometric, shapes.symbols, fit, backgrounds, calc}
\usepackage{adjustbox}
\usepackage{multirow}
\usepackage{subcaption}

\usepackage[ruled,vlined,linesnumbered]{algorithm2e}
\SetKwInput{Input}{input}
\SetKwInput{Output}{output}
\SetKwInput{Hyper}{hyper-params}
\SetKwProg{Fn}{Procedure}{}{}
\SetKwFunction{PrimalStep}{PrimalStep}
\SetKwFunction{DualStep}{DualStep}
\SetKwFunction{Substitute}{Substitute}
\SetKwFunction{Relax}{Relax}
\SetKwFunction{TopoSort}{TopoSort}
\SetKwFunction{ReverseTopoSort}{ReverseTopoSort}
\SetKwFunction{IsActivation}{IsActivation}
\SetKwFunction{IsLinear}{IsLinear}
\SetKwFunction{Predecessors}{Predecessors}
\SetKwFunction{Concretize}{Concretize}

\makeatletter

\renewcommand{\sectionautorefname}{\S\@gobble}
\renewcommand{\subsectionautorefname}{\S\@gobble}
\renewcommand{\subsubsectionautorefname}{\S\@gobble}
\def\appendixautorefname{\S\@gobble}%

\makeatother

\newtoggle{usecomment}
\settoggle{usecomment}{false}
\newcommand{\tl}[1]{\iftoggle{usecomment}{{\color{red}{[TL]: #1}}}{}}
\newcommand{\hd}[1]{\iftoggle{usecomment}{{\color{blue}{[HD]: #1}}}{}}

\newcommand{\sota}{SoTA}

\newcommand{\ibp}{\textsc{Ibp}}
\newcommand{\zonotope}{\textsc{Zonotope}}
\newcommand{\deeppoly}{\textsc{DeepPoly}}
\newcommand{\crown}{\textsc{Crown}}
\newcommand{\alphacrown}{\textsc{$\alpha$-Crown}}
\newcommand{\sdpcrown}{\textsc{Sdp-Crown}}
\newcommand{\gcpcrown}{\textsc{Gcp-Crown}}

\newcommand{\tool}{\textsc{Clad}}

\newcommand{\eg}{\emph{e.g.}}

\title{\tool{}: Constrained Abstract Domain for Neural Network Verification}

\author{Hai Duong \\
	Department of Computer Science \\
	George Mason University \\
        Fairfax, VA, USA \\
	\And
	Thanh Le \\
	Unaffiliated \\
	Yokosuka, Japan \\
	\And
	ThanhVu Nguyen \\
	Department of Computer Science \\
	George Mason University \\
        Fairfax, VA, USA \\
}

\date{}

\renewcommand{\undertitle}{}
\renewcommand{\shorttitle}{\tool{}: Constrained Abstract Domain for Neural Network Verification}

\hypersetup{
  pdftitle={CLAD: Constrained Abstract Domain for Neural Network Verification},
  pdfauthor={Hai Duong, Thanh Le, ThanhVu Nguyen},
}

\begin{document}
\maketitle

\begin{abstract}
    Neural network verification (NNV) formally verifies that a network satisfies a specified property for all inputs within a defined region.
    Modern NNV tools employ \emph{abstract domains} to compute a sound over-approximation of the network's behavior from the given input region,
    thus the tightness of these abstractions essentially determines efficiency.
    A long line of increasingly precise domains has been developed,
    but they all describe the valid input region in the same restrictive way, \eg, an $\ell_p$-norm ball.
    A practical input region is rarely a simple $\ell_p$ ball, but rather a combination $\ell_p$ ball with additional constraints.
    Verifying a network over such a region with existing abstraction produces a loose over-approximation,
    which results in either failing to verify a property or spurious counterexamples.
    We introduce \emph{\underline{C}onstrained \underline{L}agrangian \underline{A}bstract \underline{D}omain} (\tool{}),
    a new abstract domain that computes a sound over-approximation of neural networks over input regions defined by a combination of convex constraints.
    \tool{} propagates these constraints and tightens bounds over the true feasible region.
    However, bounding a neuron over the intersection of these constraints has no closed-form solution,
    so \tool{} relaxes each constraint into the objective with a Lagrange multiplier
    and solves the resulting max-min problem with a projected primal-dual method,
    alternating a projected gradient step on the input with a multiplier update.
    \tool{} supports any convex constraint with a subgradient, \eg, from automatic differentiation.
    We evaluate \tool{} on 1{,}944 instances across four convolutional networks with motion-blur structured perturbations with halfspace or $\ell_2$-ball constraints.
    On standard unconstrained $\ell_\infty$ property, \tool{} verifies as many instances as \gcpcrown{} at a similar runtime.
    On constrained properties, \tool{} verifies 60\% more instances than \gcpcrown{} on $\ell_2$-ball properties, and 22\% more in total.
\end{abstract}

\section{Introduction}
    \label{sec:intro}
    Deep neural networks (DNNs) have been increasingly adopted in real-world applications and safety-critical domains such as autonomous driving~\cite{shao2023safety}, medical diagnostics~\cite{bizjak2022deep}, and wireless communication~\cite{le2025fggm,le2026verifying}.
    However, they can fail in unexpected ways, \eg, physical obstructions like stickers on road signs can fool traffic classifiers~\cite{eykholt2018robust}.
    \emph{Neural network verification} (NNV), a new fledgling field~\citep{wang2018formal,singh2019abstract,katz2022reluplex,wang2021beta}
    aims to formally verify that a network behaves correctly under adverse conditions~\citep{huang2020survey},
    \eg, a network satisfies a specified property for all possible inputs within a defined region.
    In recent years, researchers have developed many NNV algorithms and tools
    ~\citep{chiu2025sdp,zhou2024scalable,duong2026generating,duong2025neuralsat,duong2025neuralsat2,duong2026compositional,duong2026verifying,duong2026verifying2,wu2024marabou,pyrat}.
    State-of-the-art (\sota{}) NNV tools, which adapt program analysis and constraint solving techniques,
    have been shown to be effective and scale to large networks with millions of parameters, bringing many excitements and promises to NNV research~\cite{kaulen20256th}.

    Despite recent advances, scalability remains a major challenge in NNV.
    Modern NNV tools adopt \emph{branch-and-bound} (BaB)~\cite{bunel2020branch,duong2026verifying3},
    which recursively divides a verification problem into subproblems,
    and at each subproblem, used an \emph{abstract domain}, inspired by abstract interpretation~\citep{cousot1977abstract},
    to compute an over-approximation of the network's behavior from given input region.
    \emph{The tightness of abstraction is what ultimately determines the verification efficiency}, \eg, a more precise domain prunes more subproblems before it explodes combinatorially.

    A long line of NNV work has therefore pursued computationally efficient yet precise abstract domains,
    ranging from fast but coarse intervals (\ibp{}~\citep{wang2018formal}) to relational domains such as \zonotope{}~\citep{singh2018fast}
    and the polytope-based domains, \eg, \crown{}~\citep{zhang2018efficient} and \deeppoly{}~\citep{singh2019abstract}.
    More recent refinements push this precision further,
    including optimizing relaxations (\alphacrown{}~\citep{xu2020fast}),
    tightening $\ell_2$-ball (\sdpcrown{}~\citep{chiu2025sdp}),
    or adding cutting planes on hidden neurons (\gcpcrown{}~\citep{zhang2022general}).
    However, these works are all confined to compute abstraction of one $\ell_p$-norm ball properties.

    Recent work has moved toward richer, more realistic properties, such as structural robustness against coordinated input transformations,
    \eg, filtering~\citep{duong2026verifying} or combined perturbations~\citep{duong2026verifying2}.
    However, these efforts do not enrich the abstract domain itself but \emph{reduce} the property to a standard $\ell_\infty$ problem
    by prepending an auxiliary subnetwork that generates the perturbation from fewer control variables.
    This thus restricts the property to those that can be reduced by changing the architecture of the network and cannot express arbitrary input constraints.

    The need for an abstract domain that can directly handle more expressive input regions is not hypothetical: whenever an input region is described \emph{accurately}, it is rarely a single $\ell_p$ ball.
    A growing line of work that reasons directly about the \emph{input space} of a network such as \emph{verifiable input space}
    - the largest region on which a network's outputs are certified safe -
    characterize it as a union of balls bounded by additional safety constraints~\citep{chehade2025levis}.
    Likewise, methods that compute the \emph{preimage} of a target output set
    represent it as a union and intersection of many halfspaces~\citep{zhang2024provable}.
    In both cases, the region of interest is expressed as a base $\ell_p$ ball combined with additional convex constraints.

    In this work, we propose \emph{\underline{C}onstrained \underline{L}agrangian \underline{A}bstract \underline{D}omain} (\tool{}),
    a new abstract domain that verifies DNNs over input regions defined by a combination of convex constraints.
    \tool{} propagates these constraints through the network and tightens every neuron's bounds over the true feasible region,
    \eg, the combination of all the input constraints.
    Intuitively, each added constraint rules out inputs the network can never receive,
    so \tool{} no longer has to account for them and can produce a tighter bound.
    To tackle the challenge of tightening the bounds over combined convex constraints (which has no closed form),
    \tool{} relaxes each constraint into the objective with a Lagrange multiplier and solves the resulting max-min problem with a projected primal-dual method,
    alternating a projected gradient step on the input with a multiplier update.
    Each constraint is supplied as a convex constraint function with a subgradient,
    so \tool{} supports any such constraint, \eg, the intersection of a halfspace and $\ell_2$ balls, through auto-gradient.
    The resulting bound is \emph{sound} and \emph{tight}.

    We evaluate \tool{}
    on 1{,}944 instances across four convolutional networks with motion-blur structured perturbations, each augmented with halfspace or $\ell_2$-ball constraints.
    On the standard unconstrained $\ell_\infty$ property, \tool{} verifies as many instances as \gcpcrown{} at a similar runtime.
    On constrained properties the gap widens sharply: \tool{} verifies about 60\% more instances than \gcpcrown{} on $\ell_2$-ball properties, at a higher runtime,
    because every baseline bounds over the enclosing $\ell_\infty$ interval,
    while \tool{} propagates the constraint and tightens each neuron bound over the feasible region.
    Overall, \tool{} verifies 22\% more instances than the \sota{} \gcpcrown{} and $3.3\times$ as many on the largest network,
    and its advantage grows as input perturbation strength increases.

    Our contributions include:
        (1) A new abstract domain \tool{} that handles additional convex constraints by relaxing each one into the objective with a Lagrange multiplier and computing its gradient by automatic differentiation;
        (2) We prove \tool{} returns a sound bound at every iteration, so it can stop early, and mechanize its duality results in Lean;
        (3) We establish the stopping conditions on the constraints under which \tool{} returns a tight bound; and
        (4) We evaluate \tool{} on 1,944 instances with and without halfspace and $\ell_2$-ball constraints,
        where it verifies 22\% more instances than \gcpcrown{} overall and about 60\% more on $\ell_2$-ball properties, while matching \gcpcrown{} on unconstrained $\ell_\infty$ properties at a similar runtime.

\section{Motivating Example}
    \label{sec:motivating}


    We illustrate \tool{} using a simple DNN with ReLU in~\autoref{fig:motivating-example}.
    Learned weights are shown on the edges, while learned bias for each neuron is shown above or below it.
    For illustration, we use $x$ and $x'$ neurons to represent the pre-and post-activations of the ReLU layers,
    \eg, $x'_3$ and $x'_4$ are the post-activation values after applying ReLU to $x_3$ and $x_4$, respectively.
    The computation is then:
    \begin{align}\label{eq:forward-computation}
      x_3 &= -2x_1 + 2x_2, & x'_3 &= \mathrm{ReLU}(x_3), & x_4 &= 2x_1 - x_2, & x'_4 &= \mathrm{ReLU}(x_4), \quad y= x'_5 + 2x'_6 \nonumber\\
      x_5 &= -2x'_4 + 9, & x'_5 &= \mathrm{ReLU}(x_5), & x_6 &= -x'_3 - 1, & x'_6 &= \mathrm{ReLU}(x_6)
    \end{align}
    Suppose we want to verify that $y$ is always positive for all inputs $(x_1, x_2)$ in a certain input region:
    \begin{equation}\label{eq:specs}
      \forall (x_1, x_2),\quad \underbrace{x_1 \in [-2, 2] \land x_2 \in [-1, 1]}_{\mathcal{B}} \land \underbrace{x_1 + x_2 \le 0}_{\mathcal{H}} \implies y > 0
    \end{equation}
    This property involves the interval range $\mathcal{B}: x_1 \in [-2, 2]$ and $x_2 \in [-1, 1]$
    and the half-space\footnote{A linear constraint in 2D involving two variables.} constraint  $\mathcal{H}: x_1 + x_2 \le 0$.

      To handle complex constraints, \tool{} associates each variable $v_i$ with symbolic polyhedral bounds (to preserve relational information) and concrete scalar bounds $l_i, u_i$ (to check ReLU stability).
      The key difference from prior work is \emph{how} these scalars are computed: while \deeppoly{}~\citet{singh2019abstract} evaluates the symbolic expressions merely over the input box $\mathcal{B}$, \tool{} yields tighter scalars by optimizing over the full feasible set $\mathcal{B} \cap \mathcal{H}$ (where $\mathcal{H}: x_1 + x_2 \le 0$ is the halfspace constraint).

      Starting with inputs $x_1 \in [-2, 2]$ and $x_2 \in [-1, 1]$, the first affine transformer computes:
      \begin{equation}\label{eq:aff1}
        x_3 = -2x_1+2x_2, \qquad x_4 = 2x_1-x_2
      \end{equation}
      Using the input bounds, we get concrete ranges $x_3 \in [-6,6]$ and $x_4 \in [-5,5]$. Since these ranges cross zero, both neurons are unstable. \tool{} therefore applies the standard linear relaxation (\autoref{sec:background:linear-relaxation}). For the lower bound we take $x' \ge 0$, and for the upper bound we follow \deeppoly{} and compute:
      \begin{equation}\label{eq:relu1}
        0 \le x'_3 \le 0.5\,x_3 + 3, \qquad 0 \le x'_4 \le 0.5\,x_4 + 2.5
      \end{equation}

      Propagating these abstractions through the subsequent layers yields the final output $y$. By back-substituting the intermediate symbolic bounds layer-by-layer, \tool{} expresses the lower bound of $y$ in terms of the original inputs $(x_1, x_2)$:
      \begin{equation}\label{eq:lag-lb}
        y \ge -2x_1 + x_2 + 4
      \end{equation}

        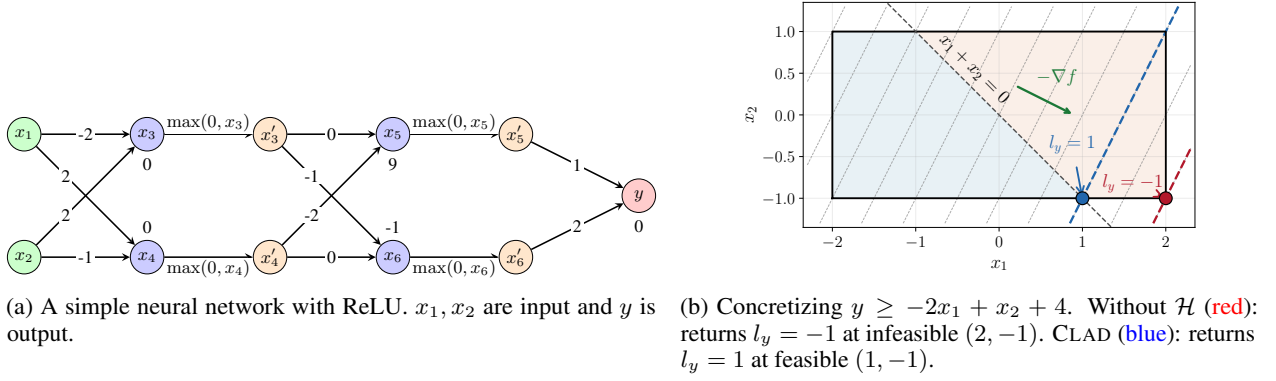
\begin{figure}[t]
            \begin{subfigure}[t]{0.52\linewidth}
            \centering
            \resizebox{\linewidth}{!}{%
            \begin{tikzpicture}[
                >=stealth,
                neuron/.style={circle, draw, minimum size=6mm, inner sep=0pt, font=\small},
                input/.style={neuron, fill=green!20},
                hidden/.style={neuron, fill=blue!20},
                relu/.style={neuron, fill=orange!20},
                output/.style={neuron, fill=red!20},
                weight/.style={font=\footnotesize, fill=white, inner sep=1pt}, 
                bias/.style={font=\footnotesize, text=black}
            ]

                \def\layersep{2.2cm} 
                \def\nodesep{2.2cm}    

                \node[input] (x1) at (0, \nodesep) {$x_1$};
                \node[input] (x2) at (0, 0) {$x_2$};

                \node[hidden, label={[bias]below:0}] (x3) at (\layersep, \nodesep) {$x_3$};
                \node[hidden, label={[bias]above:0}] (x4) at (\layersep, 0) {$x_4$};

                \node[relu] (x5) at (2*\layersep, \nodesep) {$x'_3$};
                \node[relu] (x6) at (2*\layersep, 0) {$x'_4$};

                \node[hidden, label={[bias]below:9}] (x7) at (3*\layersep, \nodesep) {$x_5$};
                \node[hidden, label={[bias]above:-1}] (x8) at (3*\layersep, 0) {$x_6$};

                \node[relu] (x9) at (4*\layersep, \nodesep) {$x'_5$};
                \node[relu] (x10) at (4*\layersep, 0) {$x'_6$};

                \node[output, label={[bias]below:0}] (y) at (5*\layersep, 0.5*\nodesep) {$y$};

                \draw[->, thick] (x1) -- node[weight] {-2} (x3);
                \draw[->, thick] (x2) -- node[weight, pos=0.3] {2} (x3);
                \draw[->, thick] (x1) -- node[weight, pos=0.3] {2} (x4);
                \draw[->, thick] (x2) -- node[weight] {-1} (x4);
                \draw[->, thick] (x3) -- node[weight, above] {$\max(0, x_3)$} (x5);
                \draw[->, thick] (x4) -- node[weight, below] {$\max(0, x_4)$} (x6);

                \draw[->, thick] (x5) -- node[weight] {0} (x7);
                \draw[->, thick] (x6) -- node[weight, pos=0.3] {-2} (x7);
                \draw[->, thick] (x5) -- node[weight, pos=0.3] {-1} (x8);
                \draw[->, thick] (x6) -- node[weight] {0} (x8);
                \draw[->, thick] (x7) -- node[weight, above] {$\max(0, x_5)$} (x9);
                \draw[->, thick] (x8) -- node[weight, below] {$\max(0, x_6)$} (x10);

                \draw[->, thick] (x9) -- node[weight] {1} (y);
                \draw[->, thick] (x10) -- node[weight] {2} (y);

            \end{tikzpicture}}%
            \caption{A simple neural network with ReLU. $x_1,x_2$ are input and $y$ is output.}
            \label{fig:motivating-example}
            \end{subfigure}\hfill
            \begin{subfigure}[t]{0.46\linewidth}
              \centering
              \includegraphics[width=0.82\linewidth]{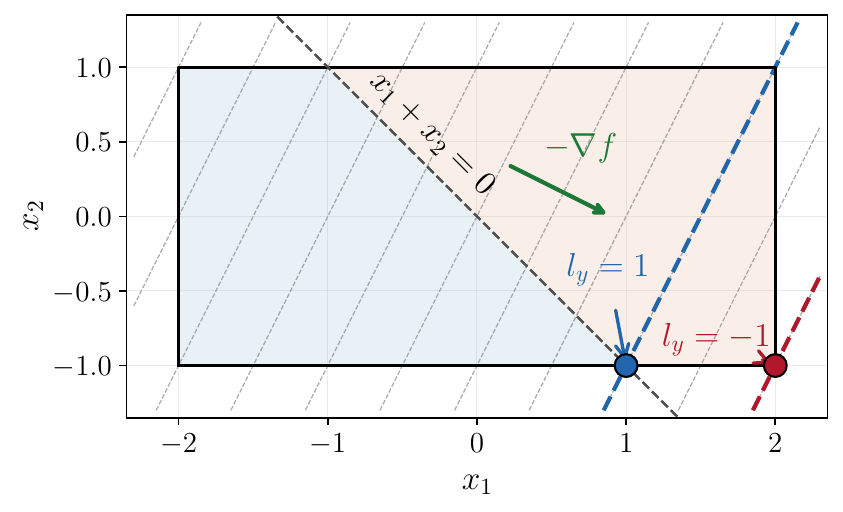}
              \caption{%
                Concretizing $y \ge -2x_1 + x_2 + 4$.
                Without $\mathcal{H}$ (\textcolor{red}{red}): returns $l_y=-1$ at infeasible $(2,-1)$.
                \tool{} (\textcolor{blue}{blue}): returns $l_y=1$ at feasible $(1,-1)$.
              }
              \label{fig:x7-concretize}
            \end{subfigure}
            \caption{Motivating example: network architecture (a) and constrained concretization (b).}
        \end{figure}

        To verify $y > 0$, \tool{} computes the concrete lower bound of~\autoref{eq:lag-lb} over the feasible input region:
        \begin{equation}\label{eq:primal-y}
            l_y \;=\; \min_{x\in\mathcal{B}} \; f(x) \;=\; \min_{x\in\mathcal{B}} \;\bigl(-2x_1 + x_2 + 4\bigr) \quad\text{subject to}\quad x_1+x_2\le 0
        \end{equation}
        To solve~\autoref{eq:primal-y}, 
        we use a Lagrangian technique~\citep{boyd2004convex,dvijotham2018dual},
        which converts the constrained optimization into an unconstrained one by folding the constraints into the objective as penalty terms
        that discourage infeasible solutions. 
        Let $\mathcal{L}(x,\mu)$ denote the Lagrangian, which contains the minimized objective $f$ and the constraint function $h(x) = x_1+x_2 \le 0$ weighted by $\mu \ge 0$:
        \begin{equation}\label{eq:lagrangian}
          \mathcal{L}(x,\mu) \;=\; (-2x_1 + x_2 + 4) \;+\; \mu\,(x_1+x_2), \qquad \mu \ge 0
        \end{equation}
        where $\mu \ge 0$ controls how strongly violations are penalized.
        For every feasible $x$, the term $\mu(x_1+x_2)$ is non-positive, so $\mathcal{L}(x,\mu) \le f(x)$ and hence $\min_{x\in\mathcal{B}}\mathcal{L}(x,\mu) \le l_y$ for every $\mu \ge 0$.
        We now solve the constrained problem~\autoref{eq:primal-y} using the Lagrangian~\autoref{eq:lagrangian}.

        In particular, if $\mu$ is too small (\eg, $\mu=0$), the penalty barely matters and the minimizer runs to infeasible point $(2,-1)$ (unconstrained solution).
        If $\mu$ is too large, the term $\mu(x_1+x_2)$ becomes a huge \emph{reward} for points deep inside the halfspace (where $x_1+x_2$ is very negative).
        The best $\mu$ sits in between, balancing the two effects.
        \tool{} searches for $\mu$ that gives the largest, and thus tightest, lower bound via an iterative loop that raises $\mu$
        proportionally to current violation until the penalty is strong enough to keep $(x_1,x_2) \in \mathcal{H}$:
        \begin{equation}\label{eq:dual-x7}
          \max_{\mu\ge 0}\;\min_{x\in\mathcal{B}} \mathcal{L}(x,\mu)  \;=\; \max_{\mu\ge 0}\;\min_{x\in\mathcal{B}} \big\{ (-2x_1 + x_2 + 4) \;+\; \mu\,(x_1+x_2) \big\}
        \end{equation}
        \tool{} solves this dual optimization problem by alternating between two steps:
        (i) adjusting $x$ to decrease the Lagrangian and (ii) adjusting $\mu$ to increase the penalty on violations until convergence.
        Each iteration performs two steps: 

        \textbf{(1) Primal step.} \tool{} moves $(x_1,x_2)$ along the negative gradient of $\mathcal{L}$, which decreases the objective $-2x_1+x_2+4$.
        The penalty $\mu(x_1+x_2)$ adds $-\mu(1,1)$ to this direction at every iterate 
        so it lowers $x_1+x_2$ and pulls a violating point 
        back toward $\mathcal{H}$ with a strength set by $\mu$.
        After the move, $(x_1,x_2)$ is clamped back into the interval $\mathcal{B}$

        \textbf{(2) Dual step.} \tool{} increases the penalty weight $\mu$ by an amount proportional to how much $\mathcal{H}$ is currently violated.
        This makes subsequent primal steps push more strongly toward the feasible region, so the two steps drive $(x_1,x_2)$ toward a point that is both low-objective and feasible.

        The solver stops once its certified bound stops improving, returning $l_y=1$, attained at $(1,-1)$ on the boundary $x_1+x_2=0$
        instead of $l_y=-1$ at infeasible $(2,-1)$.

        \textbf{Result.}
        Finally, we check the result against the specification,
        which confirms the computed lower bound $l_y = 1 > 0$ successfully verifies the specification $y > 0$ over the feasible input region $\mathcal{B}\cap\mathcal{H}$, where ignoring $\mathcal{H}$ gives the loose bound $l_y=-1$.


\section{The \tool{} Abstraction Domain}
    \label{sec:tool}

    \autoref{alg:clad-main} shows how \tool{} computes output bounds for a DNN $\mathcal{N}$ over an input ball $\mathcal{B}$ with constraints $\{h_j\}_{j=1}^{M}$.
    \tool{} first traverses $\mathcal{N}$ in topological order and computes concrete bounds at activation nodes to construct their linear relaxations (\autoref{line:alg:clad-main:cond}), while linear nodes are propagated exactly in symbolic form.
    For each node $v_i$, the abstraction maintains symbolic and concrete bounds:
    \begin{equation*}
        \underline{A}_i x+\underline{b}_i\le v_i\le\overline{A}_i x+\overline{b}_i,
        \qquad
        l_i\le v_i\le u_i.
    \end{equation*}
    The symbolic bounds preserve the dependence on the inputs, while the concrete bounds determine the activation relaxation.
    To bound each activation $v_i$, \tool{} collects its predecessors (\autoref{line:alg:clad-main:subgraph}) and initializes the symbolic bounds (\autoref{line:alg:clad-main:init}).
    It then traverses the subgraph in reversed topological order, tracing $v_i$ backward until both symbolic bounds are expressed in the inputs (\autoref{line:alg:clad-main:backward}).

    During the traversal, \tool{} substitutes the exact affine map of each linear node (\autoref{line:alg:clad-main:linear}).
    For a linear or convolution node with weight $W_k$ and bias $c_k$, \Substitute composes the exact affine:
    \begin{equation*}
        \Substitute(A,b,v_k)=(AW_k,\;Ac_k+b).
    \end{equation*}
    At each activation, \tool{} constructs lower and upper linear relaxations (\autoref{sec:background:linear-relaxation}) from its concrete bounds $[l_k,u_k]$ (\autoref{line:alg:clad-main:relubuild}).
    It selects which relaxation to substitute based on the coefficient sign (\autoref{line:alg:clad-main:reluacc}-\autoref{line:alg:clad-main:reluacc2}).
    For $r(x)=d_r\odot x+e_r$ and $r'(x)=d_{r'}\odot x+e_{r'}$, the substitution is
    \begin{equation}
        \Substitute(A,b,r,r')=\big(A^+\mathrm{diag}(d_r)+A^-\mathrm{diag}(d_{r'}),\;A^+e_r+A^-e_{r'}+b\big)
        \label{eq:subst}
    \end{equation}
    where $A^+=\max(A,0)$ and $A^-=\min(A,0)$ elementwise. Passing $(\underline{r}_k,\overline{r}_k)$ gives a sound lower bound (a positive coefficient preserves the bound direction), while passing $(\overline{r}_k,\underline{r}_k)$ gives a sound upper bound (a negative coefficient reverses the direction).
    \begin{algorithm}[t]
        \footnotesize
        \caption{\tool{} algorithm.}
        \label{alg:clad-main}
        \Input{DNN $\mathcal{N}$; input ball $\mathcal{B} = \{x : \|x - x_0\|_p \le \varepsilon\}$; constraint functions $\{h_j\}_{j=1}^{M}$}
        \Output{output bounds $[l_{\mathrm{out}}, u_{\mathrm{out}}]$}

        \For{$v_i$ in $\TopoSort(\mathcal{N})$}{
            \If{$\IsActivation(v_i) \lor v_i = v_{\mathrm{out}}$\label{line:alg:clad-main:cond}}{
                $\mathcal{G}_i \leftarrow \{v_i\} \cup \Predecessors(v_i)$ \tcp*{sub-DAG feeding $v_i$} \label{line:alg:clad-main:subgraph}
                $(\underline{A}_i, \underline{b}_i), (\overline{A}_i, \overline{b}_i) \leftarrow (I, 0), (I, 0)$ \tcp*{init to the trivial bound $v_i \le v_i \le v_i$} \label{line:alg:clad-main:init}
                \tcp{compute $v_i$ symbolically in the inputs $x$}
                \For(\tcp*[f]{back-substitute from $v_i$ to input $x$}){$v_k$ in $\ReverseTopoSort(\mathcal{G}_i)$}{ \label{line:alg:clad-main:backward}
                    \uIf(){$\IsActivation(v_k)$}{
                        $(\underline{r}_k, \overline{r}_k) \leftarrow \Relax(v_k, l_k, u_k)$ \tcp*{linear bound relaxation $\underline{r}_k \le v_k \le \overline{r}_k$} \label{line:alg:clad-main:relubuild}
                        $(\underline{A}_i,\underline{b}_i) \leftarrow \Substitute(\underline{A}_i,\underline{b}_i,\underline{r}_k,\overline{r}_k)$ \tcp*{substitute lower bound} \label{line:alg:clad-main:reluacc}
                        $(\overline{A}_i,\overline{b}_i) \leftarrow \Substitute(\overline{A}_i,\overline{b}_i,\overline{r}_k,\underline{r}_k)$ \tcp*{substitute upper bound}\label{line:alg:clad-main:reluacc2}
                    }
                    \uElse{
                        $(\underline{A}_i,\underline{b}_i) \leftarrow \Substitute(\underline{A}_i,\underline{b}_i,v_k)$ \tcp*{substitute lower bound} \label{line:alg:clad-main:linear}
                        $(\overline{A}_i,\overline{b}_i) \leftarrow \Substitute(\overline{A}_i,\overline{b}_i,v_k)$ \tcp*{substitute upper bound}
                    }
                }
                \tcp{concretize each side over $\mathcal{B} \cap \bigcap_j \{h_j \le 0\}$ (\autoref{alg:clad-pd})}
                $l_i \leftarrow \Concretize(\underline{A}_i,\underline{b}_i,\mathcal{B},\{h_j\},+1)$ \tcp*{compute concrete lower bound} \label{line:alg:clad-main:concretize_lower}
                $u_i \leftarrow \Concretize(\overline{A}_i,\overline{b}_i,\mathcal{B},\{h_j\},-1)$ \tcp*{compute concrete upper bound} \label{line:alg:clad-main:concretize_upper}
            }
        }
        \Return{$[l_{\mathrm{out}}, u_{\mathrm{out}}]$}\;
    \end{algorithm}

    Finally, \tool{} concretizes the symbolic bounds into $[l_i,u_i]$ by optimizing them over the feasible input set (\autoref{line:alg:clad-main:concretize_lower} -\autoref{line:alg:clad-main:concretize_upper}).
    This feasible set combines the input ball $\mathcal{B}$ with the additional constraints $\{h_j\}_{j=1}^{M}$ (\autoref{sec:tool:constraint}).
    \tool{} solves this optimization with the projected primal dual method (\autoref{sec:tool:concretize} and~\autoref{alg:clad-pd}).
    \tool{} uses the resulting $[l_i,u_i]$ to construct later activation relaxations and returns the output range as the final certificate.

    \subsection{Constraint Formulation}
        \label{sec:tool:constraint}

        The input property of a NNV instance can contain one or more constraints.
        Adding a constraint to \tool{} requires only a \emph{constraint function} that measures its violation.

        \noindent\textbf{Constraint.}
        Each constraint has the standard form $h_j(x)\le 0$, where $h_j:\mathbb{R}^D\to\mathbb{R}$.
        The $h_j(x)$ is nonpositive iff $x$ satisfies the constraint; otherwise, its positive value measures the violation.

        \noindent\textbf{Penalty.}
        The nonnegative penalty associated with $h_j$ is $\phi_j(x)=\max\bigl(0,h_j(x)\bigr)$, which equals zero when the constraint is satisfied and measures the violation otherwise.

        \noindent\textbf{Example.}
        For the $\ell_2$ ball $\mathcal{B}_2=\{x:\|x-x_c\|_2\le r\}$, the constraint function is $h(x)=\|x-x_c\|_2-r$ and the penalty is
        $\phi(x)=\mathrm{ReLU}\bigl(\|x-x_c\|_2-r\bigr)$.
        \tool{} uses $\phi_j$ in place of $h_j$: it is convex, zero on the feasible set, and has subgradient $0$ at its kink, so the non-differentiable center $x_c$ never enters.

    \subsection{Contrained Abstraction}
        \label{sec:tool:concretize}
        Linear relaxation methods~\citep{zhang2018efficient,singh2019abstract} bound each neuron $v_i$ in the inputs $x$ by $\underline{A}_i x+\underline{b}_i\le v_i\le\overline{A}_i x+\overline{b}_i$.
        \emph{Concretization} optimizes these expressions over all valid inputs to obtain $[l_i,u_i]$.
        These scalars determine the stability of an activation and its relaxation.
        A piecewise linear activation is \emph{stable} if $[l_i,u_i]$ lies within one linear piece, where the activation is linear and needs no approximation.
        It is \emph{unstable} if $[l_i,u_i]$ spans the nonlinear region, requiring a linear relaxation.
        Thus, tighter ranges keep more neurons stable and yield tighter bounds~\citep{xu2024training}.
        For the $\ell_p$ ball $\mathcal{B}=\{x:\|x-x_0\|_p\le\varepsilon\}$, H\"older's inequality gives the closed form:
        \begin{equation}\label{eq:holder-concretization}
          \begin{aligned}
        u_i \;=\; \overline{A}_i x_0 + \varepsilon\,\|\overline{A}_i\|_q + \overline{b}_i \qquad
        l_i \;=\; \underline{A}_i x_0 - \varepsilon\,\|\underline{A}_i\|_q + \underline{b}_i
        \end{aligned}
        \end{equation}
        where $1/p + 1/q = 1$.
        When $p=\infty$ and $q=1$, the worst case independently moves each input coordinate to the interval boundary matching the sign of its coefficient in $\overline{A}_i$.
        The result is its value at $x_0$ plus $\varepsilon$ times the sum of absolute coefficients, or the $\ell_1$ norm.
        An additional constraint such as $x_1+x_2\le0$ may remove the interval boundaries from feasible space, thus, closed-form solutions from Hölder is no longer valid.
        \tool{} concretizes each neuron over the feasible space $\mathcal{B}\cap\mathcal{H}$:
        \begin{equation}\label{eq:concretize}
          \begin{aligned}
            l_i \;=\; \min_{x \in \mathcal{B} \cap \mathcal{H}} \underline{A}_i x + \underline{b}_i \qquad
            -u_i \;=\; \min_{x \in \mathcal{B} \cap \mathcal{H}} -\big(\overline{A}_i x + \overline{b}_i\big)
          \end{aligned}
        \end{equation}
        where $\mathcal{B} = \{x : \|x - x_0\|_p \le \varepsilon\}$ be the base $\ell_p$ norm ball input property and $\mathcal{H} = \{x : h_j(x) \le 0,\ j = 1, \ldots, M\}$
        the feasible set formed by the base interval and additional constraints.

        There are three common methods to solve the constrained concretization problem~\citep{boyd2004convex}:
        (i) \emph{Dual ascent} evaluates $\mathrm{d}(\mu)$ exactly and ascends along its supergradient, but general $h_j$ does not provide the closed form inner minimization;
        (ii) \emph{Primal dual interior point methods} are accurate but solve each instance with costly Newton steps, so they do not scale to \tool{} concretization; and
        (iii) \emph{Projected primal dual methods}~\citep{arrow1958studies,nedic2009subgradient} use a first order descent step in $x$ projected onto $\mathcal{B}$ and an ascent step in $\mu$ projected onto $\{\mu\ge0\}$.

        Therefore, \tool{} uses projected primal-dual method to solve~\autoref{eq:concretize}.
        We present the procedure for a sound and tight lower bound $l_i$ and write its back-substituted bound as $Ax+b$;
        the upper bound is symmetric and minimizes $-(\overline{A}_i x+\overline{b}_i)$, as in~\autoref{eq:concretize}.

    \noindent\textbf{Lagrangian.}
        The Lagrangian~\citep{boyd2004convex,dvijotham2018dual} handles constrained optimization by replacing constraints $h_j(x)\le0$ with a weighted sum of their values in the objective.
        Each constraint has a nonnegative $\mu_j$.
        Concretely, the Lagrangian assigns a dual multiplier $\mu_j \ge 0$ to each inequality constraint $h_j(x) \le 0$ in ~\autoref{eq:concretize} and adds them into the objective:
        \begin{equation}\label{eq:clad-lagrangian}
          \mathcal{L}(x, \mu) \;=\; A x + b \;+\; \sum_{j=1}^{M} \mu_j\,h_j(x)
        \end{equation}
        For every feasible $x\in\mathcal{B}\cap\mathcal{H}$, $h_j(x)\le0$ and $\mu_j\ge0$, so each term $\mu_jh_j(x)\le0$.
        Thus, $\mathcal{L}(x,\mu)\le Ax+b$, with equality when all constraints are active.

    \noindent\textbf{Dual problem.}
        The dual function minimizes the Lagrangian over $\mathcal{B}$: $\mathrm{d}(\mu)=\min_{x\in\mathcal{B}}\mathcal{L}(x,\mu)$.
        Substituting~\autoref{eq:clad-lagrangian} gives the dual problem
        \begin{equation}
            \max_{\mu\ge0}\mathrm{d}(\mu)
            =\max_{\mu\ge0}\min_{x\in\mathcal{B}}
            \left(Ax+b+\sum_{j=1}^{M}\mu_jh_j(x)\right)
        \end{equation}
        Every $\mu\ge0$ gives a sound lower bound on $l_i$ (\autoref{thm:concretize-sound}).
        Thus, if the DNN verifier stops solving $\max_{\mu\ge0}\mathrm{d}(\mu)$ due to a time constraint, its current nonoptimal $\mu$ still gives a valid bound $\mathrm{d}(\mu)$.
        For $x_{\mu}\in\arg\min_{x\in\mathcal{B}}\mathcal{L}(x,\mu)$, the constraint values form a supergradient of the concave dual function:
        \begin{equation}
            \mathrm{d}(\mu')
            \le\mathrm{d}(\mu)
            +\sum_{j=1}^{M}h_j(x_{\mu})(\mu'_j-\mu_j)
        \end{equation}
        where $\mu'\ge0$ is other multiplier vector.
        Thus, a positive $h_j(x_{\mu})$ raises $\mu_j$ and vice versa.

    \noindent\textbf{Primal-dual solver.}
        \autoref{alg:clad-pd} handles both bounds with a sign $s$: it minimizes $s(Ax+b)$ over $\mathcal{B}\cap\mathcal{H}$,
        with $s=+1$ for the lower bound and $s=-1$ for the upper bound, whose result it negates.
        Its Lagrangian is $\mathcal{L}_s(x,\mu)=s(Ax+b)+\sum_j\mu_jh_j(x)$, which equals~\autoref{eq:clad-lagrangian} when $s=+1$.
        Each iteration of~\autoref{alg:clad-pd} performs two steps.
        The \emph{primal step} computes the gradient $q = s\cdot A + \sum_j \mu_j \nabla h_j(x)$ of $\mathcal{L}_s$ where each gradient $\nabla h_j$ is obtained via automatic differentiation and descends as $x \leftarrow x - \eta q$.
        The \emph{dual step} updates each multiplier as $\mu_j \leftarrow \max(0,\, \mu_j + \tau\, h_j(x))$, increasing $\mu_j$ when constraint $h_j$ is violated and decreasing it when satisfied, so the weight on each violated constraint grows until $x$ is feasible.
        Concretely, the two steps update:
        \begin{equation}
            \begin{aligned}
                q_t&=sA+\sum_{j=1}^{M}\mu_{j,t}\nabla h_j(x_t)
                &\widetilde{x}_{t+1}&=x_t-\eta_tq_t\\
                \mu_{j,t+1}&=\max\bigl(0,\mu_{j,t}+\tau h_j(\widetilde{x}_{t+1})\bigr)
                &x_{t+1}&=\Pi_{\mathcal{B}}(\widetilde{x}_{t+1})
            \end{aligned}
        \end{equation}
        The primal variable is initialized at the closed-form Hölder minimizer of $s\cdot Ax$ over $\mathcal{B}$ and projected back onto $\mathcal{B}$ by $\Pi_{\mathcal{B}}$ after each step.
        Hölder's inequality gives the warm start value
        \begin{equation}\label{eq:holder-warm}
            \min_{x\in\mathcal{B}}s(Ax+b)
            =s(Ax_0+b)-\varepsilon\|A\|_q,
            \qquad \frac{1}{p}+\frac{1}{q}=1
        \end{equation}
        For $p=\infty$, a minimizer is $x^{(0)}=x_0-\varepsilon\,\operatorname{sign}(sA)$.
        The primal variable is not projected onto $\mathcal{H}$ because the additional constraints control $x$ via the regularization terms in the Lagrangian.
        The projection used by the primal step is
        \begin{equation}
            \Pi_{\mathcal{B}}(z)=\operatorname*{arg\,min}_{x\in\mathcal{B}}\|x-z\|_2
        \end{equation}
        For an $\ell_\infty$ ball, it clips $z$ elementwise to the interval $[x_0-\varepsilon\mathbf{1},\,x_0+\varepsilon\mathbf{1}]$:
        \begin{equation}
            \Pi_{\mathcal{B}}(z)
            =\min\bigl(x_0+\varepsilon\mathbf{1},\,\max(x_0-\varepsilon\mathbf{1},\,z)\bigr)
        \end{equation}
        where $\min$ and $\max$ here are elementwise operations.

        \tool{} adapts $\eta$ with the Barzilai--Borwein step~\citep{barzilai1988two}, which takes larger steps in flat regions and smaller steps in sharply curved ones, avoiding manual tuning.
        With $dx_t=x_t-x_{t-1}$ and $y_t=q_t-q_{t-1}$, the Barzilai and Borwein update is
        \begin{equation}
            \eta_t=\max\left(0,\min\left(\frac{|dx_t^\top y_t|}{\|y_t\|_2^2},1\right)\right)
        \end{equation}

        The Lagrangian at an unconverged iterate is not a sound bound, since $\mathcal{L}_s(x,\mu)\ge\mathrm{d}(\mu)$.
        \tool{} instead evaluates, after every step, a \emph{certified bound} from the tangent of $\mathcal{L}_s$ at $x$, minimized over $\mathcal{B}$ as in~\autoref{eq:holder-warm}:
        \begin{equation}\label{eq:certified}
            \hat{\mathrm{d}}(x,\mu)=\mathcal{L}_s(x,\mu)+\min_{x'\in\mathcal{B}}g^\top(x'-x)=\mathcal{L}_s(x,\mu)+g^\top(x_0-x)-\varepsilon\|g\|_q
        \end{equation}
        where $g=s\cdot A+\sum_j\mu_j\nabla h_j(x)$ is a subgradient of $\mathcal{L}_s(\cdot,\mu)$.
        Each $h_j$ is convex, so $\mathcal{L}_s(\cdot,\mu)$ lies above its tangent and $\hat{\mathrm{d}}(x,\mu)\le\mathrm{d}(\mu)$, which is sound by~\autoref{thm:concretize-sound} for every iterate $x$ and every $\mu\ge0$.
        \autoref{alg:clad-pd} keeps the largest certified bound $\beta$, stops once $\beta$ has not improved by more than $\epsilon$ for $K$ consecutive iterations, and returns $s\cdot\beta$ (\autoref{line:alg:primalstep:final}).
        Stopping at any iteration is therefore safe.

    \begin{algorithm}[t]
        \small
        \caption{\tool{} Concretization Algorithm.}
        \label{alg:clad-pd}

        \Input{weight $A$ and bias $b$ of a linear bound;
               $\ell_p$ ball $\mathcal{B} = \lbrace x : \|x - x_0\|_p \le \varepsilon \rbrace$;
               constraint functions $\{h_j\}_{j=1}^{M}$;
               sign $s \in \lbrace -1,+1 \rbrace$}
        \Hyper{primal step size $\eta$,
               dual step size $\tau$,
               initial multiplier $\mu_0$, tolerance $\epsilon$, patience $K$,
               number of iterations $T$}
        \Output{lower bound if $s=+1$, upper bound if $s=-1$}

        $x \leftarrow \arg\min_{x \in \mathcal{B}}\, s\cdot A x$;\ \ $\beta \leftarrow \min_{x\in\mathcal{B}} s(Ax+b)$ \tcp*{warm start and $\mu{=}0$ bound, \autoref{eq:holder-warm}}
        $\mu_j \leftarrow \mu_0 \quad \forall j \in \{1, \ldots, M\}$ \tcp*{initialize dual multipliers}

        \For(\tcp*[f]{run optimization for $T$ iterations}){$t = 1$ \KwTo $T$}{
            $q, x' \leftarrow$ \PrimalStep{$x, \mu, A, \{h_j\}, \eta$} \tcp*[f]{primal descent}


            $dx \leftarrow x' - x$;\ \ $y \leftarrow \big(s \cdot A + \sum_j \mu_j \nabla h_j(x')\big) - q$ \tcp*{changes in $x$ and gradient}

            $x \leftarrow x'$ \tcp*{update primal variable}

            $\mu \leftarrow$ \DualStep{$x, \mu, \{h_j\}, \tau$} \tcp*[f]{dual ascent}

            $\beta \leftarrow \max\big(\beta,\ \hat{\mathrm{d}}(x,\mu)\big)$ \tcp*{certified bound, \autoref{eq:certified}}

            \lIf(\tcp*[f]{Barzilai-Borwein step size}){$\|y\|_2^2 > \epsilon$}{
              $\eta \leftarrow \max\big(0,\; \min\big(|dx^\top y| \,/\, \|y\|_2^2,\; 1\big)\big)$
            }

            $x \leftarrow \Pi_{\mathcal{B}} \big(x \big)$ \label{line:alg:primalstep:project}  \tcp*{project onto $\mathcal{B}$ only}

            \lIf(\tcp*[f]{stop if converged}){$\beta$ has not improved by $\epsilon$ for $K$ iterations}{
              \textbf{break}
            }
        }
        \Return{$s\cdot\beta$}\label{line:alg:primalstep:final}\tcp*{certified lower/upper bound}
    \end{algorithm}

    \subsection{Soundness, Tightness, and Convergence of \tool{}}
    \label{sec:tool:soundness}

    Every certified bound $\hat{\mathrm{d}}(x,\mu)$ is sound, so~\autoref{alg:clad-pd} can stop at any iteration and return a valid bound.
    Under the Slater condition, which requires a strictly feasible interior point, maximizing the dual function over $\mu\ge0$ recovers the exact constrained optimum and is therefore \emph{tight}.
    \autoref{sec:appendix:analysis} gives formal proofs of \autoref{thm:concretize-sound}, \autoref{thm:clad-tight}, and resulting corollaries for halfspace and $\ell_2$ ball constraints.
    For a simultaneous variant with a constant shared step and bounded multipliers, the averaged iterates converge to the constrained optimum at rate $O(1/\sqrt{T})$ (\autoref{thm:clad-convergence}); this concerns tightness only, as soundness holds at every iteration.

\section{Evaluation}
    \label{sec:eval}


    \noindent\textbf{Verification Benchmark}
        We use four convolutional networks from~\citet{chiu2025sdp}: ConvSmall (M) trained on MNIST, and ConvSmall (C), ConvDeep, and ConvLarge trained on CIFAR-10.
        We do not use VNN-COMP instances, which are generally easy
        \footnote{Out of 2700 instances in regular track of VNN-COMP'24~\citep{brix2024fifth},
        nearly 90\% problems ($\sim$2400/2700) were solved with \deeppoly{} or \alphacrown{}.}
        for this comparison: most are already solved by simple abstractions such as \ibp{} or \zonotope{}, so they cannot differentiate among more precise methods.
        In contrast, these networks span 4 to 7 layers, 8 to 64 convolutional channels, and 55K to 2.47M parameters, and they are known to separate verifiers sharply~\citep{wang2021beta,chiu2025sdp},
        \eg, on the CIFAR-10 ConvLarge, \alphacrown{}, \gcpcrown{}, and \sdpcrown{} verify 2.5\%, 6\%, and 63.5\% of the images, respectively~\citep{chiu2025sdp}.

        We generate more expressive instances using motion-blur perturbations from VeriDou~\citep{duong2026verifying2},
        whose convolutional parameterization covers a continuous range of blur angles, at six strengths for each network.
        We use structured perturbations
        noise because real corruptions such as blur change correlated pixels~\citep{duong2026verifying,duong2026verifying2}.
        For every network and perturbation strength, we compare three setups: (i) $\ell_\infty$ interval (\texttt{linf}), (ii) that same interval intersects with a halfspace (\texttt{linf+hs}), (iii) interval intersects with an $\ell_2$ ball (\texttt{linf+l2}).
        This yields 1{,}944 instances across the four networks and three property types.
        \autoref{sec:benchmark} gives more details on the benchmarks.

    \noindent\textbf{Comparison Baselines}
        We compare \tool{} with abstraction baselines including:
        \ibp{}~\citep{wang2018formal}, 
        \zonotope{}, 
        \crown{}~\citep{xu2020automatic}/\deeppoly{}~\citep{singh2019abstract}, 
        \alphacrown{}~\citep{xu2020fast}, 
        \sdpcrown{}~\citep{chiu2025sdp}, 
        and \gcpcrown{}~\citep{zhang2022general}. 
        For a fair comparison, we give every baseline the minimum interval that contain the constrained input region of each property.
        \tool{} receives the same tightened interval in addition to the constraint itself.

    \paragraph{Comparison Metrics}
        The \emph{number of verified instances} is our primary measure of an abstraction's effectiveness, since more verified instances means robustness holds over a broader range of instances.
        An instance is verified if its output lower bound $l_{\mathrm{out}} > 0$.
        For classification models, $l_{\mathrm{out}} > 0$ when the margin $Y_i - Y_j > 0$ between the true class $i$ and a competing class $j$.



    \subsection{RQ1: Performance with Existing Abstraction}
        \label{sec:rq1}

        \begin{table}[t]
          \centering
          \caption{Verified instances and average runtime per instance per network and method.}
          \label{tab:rq1}
          \setlength{\tabcolsep}{4pt}
          \resizebox{\linewidth}{!}{%
            \begin{tabular}{@{}l*{14}{r}@{}}
              \toprule
              \multirow{2}{*}{\textbf{Network}} &
              \multicolumn{2}{c}{\ibp{}} &
              \multicolumn{2}{c}{\zonotope{}} &
              \multicolumn{2}{c}{\crown{} / \deeppoly{}} &
              \multicolumn{2}{c}{\alphacrown{}} &
              \multicolumn{2}{c}{\sdpcrown{}} &
              \multicolumn{2}{c}{\gcpcrown{}} &
              \multicolumn{2}{c}{\tool{}} \\
              & Verified & Time & Verified & Time & Verified & Time & Verified & Time & Verified & Time & Verified & Time & Verified & Time \\
              \midrule
              ConvSmall (M) & 0 & 0.00 & 216 & 0.00 & 264 & 0.04 & 297 & 4.94 & 297 & 14.24 & 300 & 18.81 & \textbf{358} & 19.23 \\
              ConvSmall (C) & 0 & 0.00 & 222 & 0.00 & 270 & 0.04 & 312 & 5.24 & 312 & 15.59 & 315 & 28.77 & \textbf{375} & 22.55 \\
              ConvDeep & 0 & 0.00 & 198 & 0.00 & 243 & 0.06 & 300 & 8.35 & 303 & 22.64 & 306 & 35.95 & \textbf{354} & 38.77 \\
              ConvLarge & 0 & 0.00 & 0 & 0.10 & 0 & 0.20 & 18 & 30.78 & 18 & 63.70 & 18 & 72.37 & \textbf{59} & 78.05 \\
              \bottomrule
            \end{tabular}}%
        \end{table}

    \autoref{tab:rq1} presents the verified instance counts and average per-instance runtimes for each method across the four convolutional networks.
    \tool{} achieves the highest number of verified instances on each network, totaling 1,146, which improves 22\% compared to the strongest baseline \gcpcrown{}.
    The performance gap is most pronounced on ConvLarge, where \tool{} verifies 59 instances, whereas \alphacrown{}, \sdpcrown{}, and \gcpcrown{} each verify fewer than one-third as many instances.
    \ibp{} does not verify any instances, which indicates that the benchmark instances are non-trivial.

    In terms of runtime, \ibp{}, \zonotope{} and \crown{} require at most 0.2 seconds per instance but verify substantially fewer instances.
    while \alphacrown{}, \sdpcrown{}, and \gcpcrown{} take up to 72s per instance on ConvLarge.
    \tool{}'s runtime is close to that of \gcpcrown{} while verifying 41 to 60 more instances per network.
    Consequently, the primal-dual solver increases verification precision while maintaining a runtime similar to that of the baseline abstraction domains.

    \subsection{RQ2: Performance on Constrained and Unconstrained Properties}
        \label{sec:rq2}

        \begin{figure}[t]
          \centering
          \includegraphics[width=\linewidth]{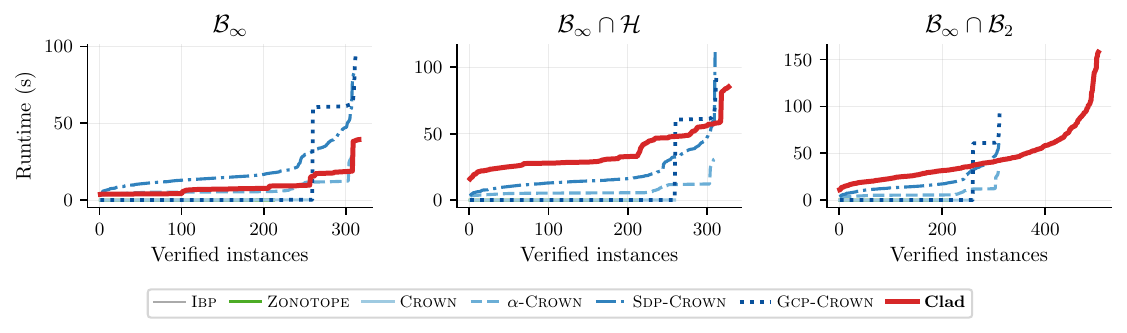}
            \caption{Cactus plots of runtime per verified instances for each property type}
          \label{fig:rq2}
        \end{figure}
        The leftmost cactus plot in \autoref{fig:rq2} presents the runtime for verified instances of the unconstrained $\ell_\infty$ property.
        \alphacrown{}, \sdpcrown{}, \gcpcrown{}, and \tool{} verify 309, 310, 313, and 316 instances, respectively.
        \gcpcrown{} verifies most instances in approximately 0.1 seconds,
        but requires around 100 seconds for 50 challenging instances when the general cutting mechanism is activated.
        \tool{} processes each instance in approximately 4 to 40 seconds.
        \tool{} therefore achieves precision comparable to that of \alphacrown{} variants, with a similar per-instance compute time.

        The advantage of \tool{} becomes pronounced when the input region includes an additional constraint
        as demonstrated by the second and third cactus plots in \autoref{fig:rq2}.
        Each baseline curve remains consistent across all three panels
        so even with a tightened interval no baseline verifies more instances than under the unconstrained property.
        On $\mathcal{B}_\infty \cap \mathcal{H}$, \tool{} verifies approximately 330 instances,
        compared to approximately 313 for the strongest baselines.
        On $\mathcal{B}_\infty \cap \mathcal{B}_2$, \tool{} verifies approximately 500 instances,
        representing a sharp increase of roughly 60\% over \gcpcrown{}.
        With additional contrains, \tool{} requires additional time for its Lagrangian concretization step,
        resulting in its cactus curve lying slightly above the baselines.
        The $\ell_2$ ball reduces the feasible region most significantly and therefore demonstrates the clearest benefit -
        \tool{} requires at most approximately 60\% time per instance while also verifies 60\% more instances than \gcpcrown{}.

    \subsection{RQ3: Performance on Different Perturbation Radii}
        \label{sec:rq3}
        \begin{figure}[t]
          \centering
          \includegraphics[width=\linewidth]{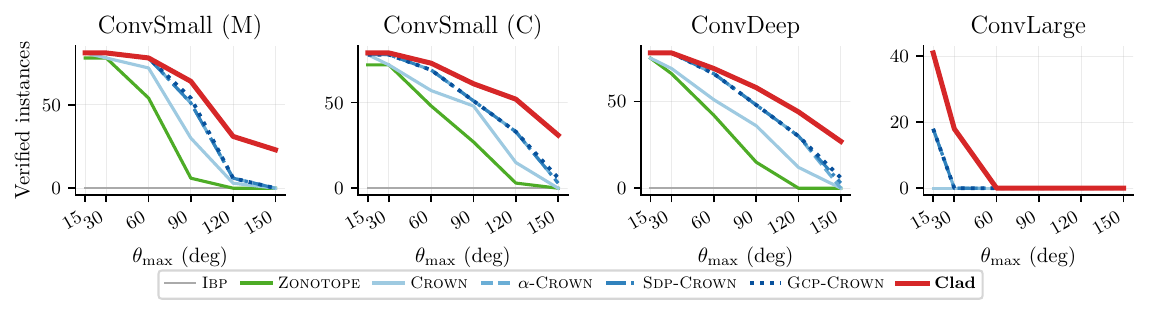}
            \caption{Verified instances per motion-blur angle $\theta_{\max}$ on the four convolutional networks.
            }
          \label{fig:rq3}
        \end{figure}

        \autoref{fig:rq3} presents the number of verified instances for each method as the maximum motion-blur angle $\theta_{\max}$ increases.
        At $15^\circ$ and $30^\circ$ on ConvSmall and ConvDeep, \alphacrown{} variants is the same as \tool{}
        because the small input space remains manageable for all these abstractions.
        As $\theta_{\max}$ increases and the input space expands, the performance gap between \tool{} and the baseline methods becomes more pronounced.
        At $150^\circ$ on ConvSmall (C) and ConvDeep, \gcpcrown{} verifies only six instances each,
        whereas \tool{} verifies 31 and 27 instances, respectively.
        ConvLarge, which contains 2.47M parameters, presents greater difficulty.
        \alphacrown{}, \sdpcrown{}, and \gcpcrown{} each verify 18 instances at $15^\circ$ and none at larger angles, whereas \tool{} verifies 41 instances at $15^\circ$ and 18 at $30^\circ$.
        Therefore, the advantage of \tool{} observed in \autoref{sec:rq2} increases as the perturbation strength and model size grows.

\section{Conclusion}
    \label{sec:conclusion}
    We presented \tool{}, an abstract domain that retains additional input constraints in a Lagrangian and uses projected primal dual concretization to tighten neuron bounds.
    \tool{} handles convex constraints and their intersections, returns a sound bound at every iteration, and is tight under the Slater condition.
    Future work includes supporting other activations, integration with branch and bound~\citep{bunel2020branch}, and extending to nonconvex perturbations~\citep{duong2026verifying}.

\bibliographystyle{unsrtnat}
\bibliography{paper}

\appendix

\section{Background}
    \label{sec:background}

    \subsection{Neural Network Verification (NNV)}


        \begin{definition}[NNV Problem]
        Given a DNN \(N\) and a property (or specification) $\phi$, the NNV problem asks whether $\phi$ is a valid property of $N$.
        Typically, $\phi$ takes the form $\phi_{in} \Rightarrow \phi_{out}$, where $\phi_{in}$ is a property over the inputs of $N$ and $\phi_{out}$ is a property over the outputs of $N$.
        \end{definition}

        Modern DNN techniques~\citep{wang2021beta,duong2024harnessing,duong2023dpll,duong2025neuralsat,zhang2022general,zhou2024scalable,bak2021nnenum,wu2024marabou,ferrari2022complete} treat this verification problem as a satisfiability problem by encoding the DNN \(N\) and the property $\phi$ as a logical formula:
        \begin{equation}
            N \land \phi_{in} \land \neg \phi_{out}
            \label{eq:dnn_verification}
        \end{equation}
        If~\autoref{eq:dnn_verification} is unsatisfiable (UNSAT), the considered property holds.
        Otherwise, it is satisfiable (SAT) and a counterexample exists that disproves the property.

    \subsection{Linear Relaxation}
        \label{sec:background:linear-relaxation}

        Linear relaxation bounds each intermediate neuron $v_i$ by affine functions of the network inputs $x$:
        \begin{equation}
            \underline{A}_i x + \underline{b}_i \;\le\; v_i \;\le\; \overline{A}_i x + \overline{b}_i
            \label{eq:linear-bound}
        \end{equation}
        For linear layers, bounds propagate exactly.
        For an unstable ReLU with pre-activation $\hat{z} \in [l, u]$, $l < 0 < u$, and output $z = \mathrm{ReLU}(\hat{z})$, \crown{}~\citep{xu2020automatic} and \deeppoly{}~\citep{singh2019abstract} over-approximate the activation by:
        \begin{equation}
          \alpha\,\hat{z} \;\le\; z \;\le\; \bar{d}\,\hat{z} + \bar{b} \qquad \bar{d} = \tfrac{u}{u-l} \qquad \bar{b} = -\tfrac{lu}{u-l} \qquad  \alpha \in [0,1]
            \label{eq:relu-upper}
        \end{equation}
        Stable neurons ($l \ge 0$ or $u \le 0$) need no relaxation.
        These per-neuron relaxations are back-substituted through preceding linear layers until~\autoref{eq:linear-bound} is expressed entirely in terms of $x$, and then \emph{concretized} by optimizing over input space $\mathcal{B}$:
        \begin{equation}
          \label{eq:concretization}
          l_i = \min_{x \in \mathcal{B}}\;(\underline{A}_i x + \underline{b}_i), \qquad u_i = \max_{x \in \mathcal{B}}\;(\overline{A}_i x + \overline{b}_i)
        \end{equation}
        For a standard $\ell_\infty$ ball, these reduce to closed-form expressions via H\"{o}lder's inequality~\cite{wang2021beta}.
        ~\hd{cite}~\tl{cited}

\section{Related Work}
    \label{sec:related}
    Existing abstractions form a spectrum of speed-accuracy tradeoffs,
    \eg, better speed by \ibp{}~\citep{wang2018formal}, \zonotope{}~\citep{singh2018fast}, \crown{}~\citep{zhang2018efficient} and \deeppoly{}~\citep{singh2019abstract},
    or better precision by \alphacrown{}~\citep{xu2020automatic}, \gcpcrown{}~\citep{zhang2022general}.
    All of these methods, however, compute bounds by optimizing over the enclosing $\ell_\infty$ interval 
    and discard any additional input constraints such as halfspace or $\ell_2$-ball intersections.
    \tool{} addresses this gap by incorporating such constraints directly into each bound propagation step.

    Lagrangian duality itself has a long history in NNV.
    \citet{dvijotham2018dual} dualize the layer equations of the entire network and optimize the multipliers by subgradient ascent, 
    \citet{bunel2020lagrangian} instead dualize the standard LP relaxation 
    which \citet{depalma2021improved} then integrate into a complete branch-and-bound verifier.
    Because these dualize the network itself 
    the dual dimension grows with the network size which isn't as scalable as linear relaxation abstraction domains 
    since linear relaxation abstractions keep closed-form bound propagation and dualize only a small set of side constraints,
    \eg, $\beta$-\textsc{Crown} for branch-and-bound splits~\citep{wang2021beta},
    \gcpcrown{} for cutting planes~\citep{zhang2022gcpcrown},
    and \textsc{Invprop} for output constraints~\citep{kotha2024provably}.
    Closest to \tool{}, Clip-and-Verify~\citep{zhou2025clipandverify} derives linear constraints in the input space from hidden-neuron branch-and-bound splits and output specifications,
    and uses them to clip the input box and retighten intermediate bounds.
    For a single linear constraint over a box, it solves the same one-dimensional piecewise-linear dual that \tool{} uses for a halfspace.
    \tool{} differs in the constraints it handles: they are part of the input specification rather than derived during verification,
    and they may be any convex function, such as an $\ell_2$ ball, whereas Clip-and-Verify handles only linear inequalities.
    \tool{}
    dualizes only the $M$ input constraints during concretization (\autoref{alg:clad-pd}), thus, yields a dual of dimension $M$ independent of the network size, which also remains scalable.

\section{Analysis of \tool{}}
    \label{sec:appendix:analysis}

    \subsection{Soundness and Tightness}

    The results in this appendix are stated for the upper bound $u_i$ in the maximization form, which matches the Lean mechanization.
    The lower-bound form of~\autoref{sec:tool:concretize} follows by replacing $Ax+b$ with $-(Ax+b)$.

    An abstraction is \emph{sound} if the bounds it returns always over-approximate and never underestimate the true range of a neuron.
    This means the value the abstraction returned for the upper bound must never fall below the true constrained maximum
    and the lower bound never reach above the true constrained minimum.
    For \tool{}, soundness holds for \emph{every} Lagrange multiplier value $\mu$,
    and every certified bound $\hat{\mathrm{d}}(x,\mu)$ of~\autoref{eq:certified} is bounded by $\mathrm{d}(\mu)$,
    which means~\autoref{alg:clad-pd} may stop at any iteration and still return a valid bound.

    \begin{theorem}[\tool{} soundness]\label{thm:concretize-sound}
    For every $\mu \ge 0$,
    \begin{equation}\label{eq:clad-weakduality}
      \mathrm{d}(\mu) \;=\; \max_{x \in \mathcal{B}} \mathcal{L}(x, \mu)
      \;\ge\; \max_{x \in \mathcal{B} \cap \mathcal{H}} A x + b \;=\; u_i
    \end{equation}
    \end{theorem}
    \begin{proof}
        At every feasible point $x \in \mathcal{B} \cap \mathcal{H}$, each constraint function satisfies $h_j(x) \le 0$.
        Thus, every penalty term $-\mu_j h_j(x)$ is non-negative and the Lagrangian dominates the objective: $\mathcal{L}(x, \mu) = A x + b -\mu_j h_j(x) \ge A x + b$.
        Maximizing over the larger set $\mathcal{B} \supseteq \mathcal{B} \cap \mathcal{H}$ can only increase this value, giving $\mathrm{d}(\mu) \ge u_i$.
    \end{proof}

    \begin{remark}[Lean mechanization]
    \label{rem:lean-soundness}
    \autoref{thm:concretize-sound} is fully mechanized in Lean~4 against Mathlib as \texttt{CLAD.primalVal\_le\_dualFn} in \texttt{lean/CLAD/Soundness.lean} of our released code.
    Lean mechanizes the weak-duality step $\mathrm{d}(\mu)\ge u_i$ for the exact dual function.
    The tangent step relating the certified bound of~\autoref{eq:certified} to $\mathrm{d}(\mu)$ is the one-line convexity argument of~\autoref{sec:tool:concretize}; it is not mechanized, and neither is floating-point arithmetic.
    Every mechanized result in this subsection is \texttt{sorry}-free and depends only on Lean's three standard axioms, \eg, \texttt{propext}, \texttt{Classical.choice} and \texttt{Quot.sound}.
    The mechanization is more general than the statement above: the first proof step, \texttt{CLAD.le\_lagrangian\_of\_mem\_feasible}, establishes $A x + b \le \mathcal{L}(x, \mu)$ at every feasible $x$ using only $\mu \ge 0$ and $h_j(x) \le 0$, with no convexity, no continuity, and no linearity of the objective.
    The mechanized statement carries two side conditions, $\mathcal{B} \cap \mathcal{H} \neq \emptyset$ and $\mathcal{L}(\cdot, \mu)$ bounded above on $\mathcal{B}$; these are artifacts of encoding $\max$ as a real-valued supremum, not mathematical content, and both hold automatically for the compact $\mathcal{B}$ of our setting.
    \end{remark}

    The condition required for \tool{} to achieve tightness follows the classical \emph{Slater condition}~\cite[\S 5.2]{boyd2004convex}.

    \begin{definition}[Slater condition]\label{def:slater}
    Assume the $M$ convex constraint functions can be ordered such that $h_1, \ldots, h_k$ are non-affine and $h_{k+1}, \ldots, h_M$ are affine.
    The Slater condition holds for the constrained problem $\max_{x \in \mathcal{B} \cap \mathcal{H}} A x + b$ if there exists a point $\hat{x}$ in the interior of the input domain $\operatorname{int} \mathcal{B}$
    such that $h_j(\hat{x}) < 0$ for every $j \in \{1, \ldots, k\}$ and $h_j(\hat{x}) \le 0$ for every $j \in \{k+1, \ldots, M\}$.
    \end{definition}

    \begin{theorem}[Tightness of \tool{} for convex constraint intersections]\label{thm:clad-tight}
    Let the radius $\varepsilon > 0$, let each constraint function $h_j$ be convex and continuous on $\mathbb{R}^D$,
    and suppose there exists a point in the relative interior of the base interval $\hat{x} \in \operatorname{int} \mathcal{B}$
    satisfying $h_j(\hat{x}) < 0$ for every non-affine constraint function $h_j(.), \forall j \in \{1, \ldots, k\}$
    and $h_{j}(\hat{x}) \le 0$ for every affine constraint function $h_{j}, \forall j \in \{k+1, \ldots, M\}$.
    Then the dual bound $\mathrm{d}(\mu)$ attains the minimal upper bound $u_i^{\star}$ at its minimum over the multiplier $\mu \ge 0$:
    \begin{equation}\label{eq:clad-strongduality}
        \min_{\mu \ge 0} \mathrm{d}(\mu) \;=\; \max_{x \in \mathcal{B} \cap \mathcal{H}} A x + b \;=\; u_i^{\star}
    \end{equation}
    and this minimum is attained at some optimal multiplier $\mu^{\star} \ge 0$.
    \end{theorem}
    \begin{proof}
        The assumptions fit Slater's condition in~\autoref{def:slater}.
        Thus, strong duality holds, the duality gap is zero, and there exists some multiplier $\mu^{\star} \ge 0$ attaining this optimum.
    \end{proof}

    \begin{remark}[Lean mechanization]
    \label{rem:lean-tightness}
    \autoref{thm:clad-tight} is mechanized as \texttt{CLAD.dualFn\_eq\_primalVal\_of\_slater} in \texttt{lean/CLAD/Tightness.lean}.
    The mechanized proof routes through Sion's minimax theorem (\texttt{Sion.exists\_isSaddlePointOn} in Mathlib) applied to $\mathcal{L}$ on $[0,\kappa]^M \times \mathcal{B}$, where the Slater margin supplies an a priori bound $\kappa$ on any multiplier whose dual value is at most $\mathrm{d}(0)$.
    The conclusion is the full claim of~\eqref{eq:clad-strongduality}: both the zero duality gap and the attainment of the minimum over \emph{all} $\mu \ge 0$, not merely over the truncation used inside the proof.
    Two differences from the statement above are worth recording.
    First, the objective is assumed only concave and continuous on $\mathcal{B}$; the affine $A x + b$ is a special case, so the mechanized result is a generalization, and the unused hypothesis $\varepsilon > 0$ is dropped.
    Second, the mechanized Slater hypothesis asks for a uniform strict margin $h_j(\hat{x}) \le -\delta < 0$ at \emph{every} constraint, \eg, the case $k = M$ of~\autoref{def:slater}; over finitely many constraints this is equivalent to requiring $h_j(\hat{x}) < 0$ for each $j$.
    The affine relaxation of~\autoref{def:slater} is therefore not used in the general theorem, but it is recovered exactly where the paper applies it, in~\autoref{cor:clad-tight-hs}.
    \end{remark}

    \autoref{thm:clad-tight} applies directly to the two constraint families of~\autoref{sec:tool:constraint}:

    \begin{corollary}[Halfspace tightness]\label{cor:clad-tight-hs}
    Let $\mathcal{H} = \{x : c x \le d\}$ with $c \neq 0$.
    If $\operatorname{int} \mathcal{B} \cap \mathcal{H} \neq \emptyset$, then $\min_{\mu \ge 0} \mathrm{d}(\mu) = u_i^{\star}$.
    \end{corollary}
    \begin{proof}
    The constraint function $h(x) = c x - d$ is affine, so feasibility of $\hat{x}$ satisfies the Slater condition.
    \end{proof}

    \begin{remark}[Lean mechanization]
    \label{rem:lean-halfspace}
    \autoref{cor:clad-tight-hs} is mechanized as \texttt{CLAD.tight\_halfspace} in \texttt{lean/CLAD/Tightness.lean}, under exactly the non-strict hypothesis $c\hat{x} \le d$ stated above.
    The strict margin demanded by \texttt{CLAD.dualFn\_eq\_primalVal\_of\_slater} is manufactured rather than assumed: \texttt{CLAD.exists\_strict\_slater\_of\_mem\_interior} perturbs $\hat{x}$ to $\hat{x} - t v$ for a direction $v$ with $c v > 0$ and small $t > 0$, which remains in $\mathcal{B}$ because $\hat{x} \in \operatorname{int} \mathcal{B}$ and strictly decreases $c x - d$ because $c \neq 0$.
    This is the step that recovers the affine relaxation of~\autoref{def:slater} discussed in~\autoref{rem:lean-tightness}.
    \end{remark}

    \begin{corollary}[$\ell_2$-ball tightness]\label{cor:clad-tight-l2}
    Let $\mathcal{H} = \{x : \|x - x_c\|_2 \le r\}$ with $r > 0$.
    If there exists $\hat{x} \in \operatorname{int} \mathcal{B}$ with $\|\hat{x} - x_c\|_2 < r$, then $\min_{\mu \ge 0} \mathrm{d}(\mu) = u_i^{\star}$.
    \end{corollary}
    \begin{proof}
    The point $\hat{x}$ is a strictly feasible Slater point for $h(x) = \|x - x_c\|_2 - r$.
    \end{proof}

    \begin{remark}[Lean mechanization]
    \label{rem:lean-l2}
    \autoref{cor:clad-tight-l2} is mechanized as \texttt{CLAD.tight\_l2Ball} in \texttt{lean/CLAD/Tightness.lean}, instantiating the Slater margin as $\delta = r - \|\hat{x} - x_c\|_2$, which is positive precisely by the hypothesis $\|\hat{x} - x_c\|_2 < r$.
    Convexity of the constraint is discharged from Mathlib's \texttt{convexOn\_dist}.
    \end{remark}

    \subsection{Convergence}
        \label{sec:appendix:convergence}

        \begin{theorem}[Convergence of the primal-dual concretization]\label{thm:clad-convergence}
        Let each constraint function $h_j$ be convex and continuous, and suppose the Slater condition of~\autoref{thm:clad-tight} holds, so that the interval $[0,\kappa]^M$ contains the optimal multiplier $\mu^{\star}$.
        With the constant step size $\alpha = R_{\mathcal{Z}} / (G \sqrt{T})$ shared by the primal and dual updates of the simultaneous variant of~\autoref{alg:clad-pd}, the iterate averages $\bar{x} = \frac{1}{T}\sum_{t \le T} x_t$ and $\bar{\mu} = \frac{1}{T}\sum_{t \le T} \mu_t$ satisfy
        \begin{equation*}
          \mathrm{d}(\bar{\mu}) - u_i^{\star} \;\le\; \frac{R_{\mathcal{Z}}\, G}{\sqrt{T}},
        \end{equation*}
        where $R_{\mathcal{Z}} = \sqrt{R^2 + \kappa^2 M}$ is the diameter of $\mathcal{Z} = \mathcal{B} \times [0,\kappa]^M$;
        $G_x = \sup_{x,\mu} \|A - \sum_j \mu_j \nabla h_j(x)\|_2$, $G_\mu = \sup_{x \in \mathcal{B}} \|(h_1(x), \ldots, h_M(x))\|_2$, and $G = \sqrt{G_x^2 + G_\mu^2}$.
        \end{theorem}
        \begin{proof}
        The bound follows by adding two regret guarantees: the primal step matches the standard rate for supergradient ascent on the concave function $\mathcal{L}(\cdot, \mu_t)$, and the dual step matches the analogous rate for subgradient descent on $\mathcal{L}(x_t, \cdot)$~\citep{nedic2009subgradient,nemirovski2009robust}; the two regrets add because their common term $\mathcal{L}(x_t, \mu_t)$ cancels.
        Averaging via Jensen's inequality and identifying the saddle value with $u_i^{\star}$ using~\autoref{thm:concretize-sound} and~\autoref{thm:clad-tight} yields the stated rate.
        \end{proof}
        This guarantee concerns an idealized simultaneous variant.
        \autoref{alg:clad-pd} uses sequential updates, the Barzilai--Borwein step, unbounded multipliers, and early stopping; its soundness rests on the certified bound of~\autoref{eq:certified}, not on convergence.

\section{Implementation Details}
    \label{sec:appendix:implementation}

    \tool{} is implemented in Python and uses the PyTorch framework~\citep{paszke2019pytorch}.
    It accepts neural networks in ONNX format~\citep{onnx} and verification properties in VNNLib format.
    The ONNX network is parsed into a directed acyclic graph, after which the bound-propagation pass from~\autoref{alg:clad-main} is executed.
    Each constraint is implemented in hinge form $\phi_j=\max(0,h_j)$.
    Its subgradient is computed in closed form for the interval and halfspace constraints and with PyTorch's autograd for the $\ell_2$ ball.

    \tool{} supports general convex constraints, and we consider
    three representative constraint types in our benchmark (\autoref{sec:benchmark}):
    halfspace ($c_j x \le d_j$, constraint function $c_j x - d_j$), $\ell_2$ ball ($\|x-x_c\|_2 \le r$,
    constraint function $\|x-x_c\|_2 - r$), and $\ell_\infty$ interval (represented directly as variable bounds during the primal projection step).
    The primal-dual solver for constraint-aware concretization (\autoref{alg:clad-pd}) is executed for a maximum of $T{=}500$ iterations,
    with a primal step size of $\eta{=}0.01$ and a dual step size of $\tau{=}0.1$.
    Multipliers start at $\mu_0=\|A\|_2$, the norm of the bound's coefficient vector.
    Besides $\hat{\mathrm{d}}(x,\mu)$ at every iterate, \tool{} folds into $\beta$ the $\mu{=}0$ bound and, for each constraint, the exact maximizer of the certified bound over that constraint's multiplier alone; all are certified bounds, so their maximum is returned.
    Early termination occurs when $\beta$ has not improved by more than $\epsilon{=}3 \times 10^{-4}$ for $K{=}10$ consecutive iterations.
    All computations use floating point without directed rounding, as in the baselines.
    \gcpcrown{} runs \alphacrown{} followed by CPLEX cut generation with 100 cuts and a 120-second budget per instance.

  \subsection{Runtime and Environment}
    Our experiments were conducted on a Linux system equipped with an AMD 12-core CPU, 128 GB of RAM, and an NVIDIA GeForce RTX 4080 GPU (16 GB VRAM),
    which is leveraged by all methods (\tool{} and baselines).
    Runtime is the average wall-clock time, in seconds, per instance.
    Because every method computes bounds in a single propagation pass (\autoref{sec:eval}), all instances run to completion and no timeouts occur.
    Runtime therefore reflects the full computational cost of each abstraction on an instance.

  \subsection{Verification Benchmark}
    \label{sec:benchmark}
    Our evaluation uses the standard $\ell_\infty$-interval verification setting and network architectures commonly used in NNV evaluation
    (\eg, VNN-COMPs~\citep{kaulen20256th,brix2024fifth,brix2023fourth}).
    We therefore generate instances, with the \emph{same} networks,
    at harder perturbation radii where the choice of abstract domain is decisive.

    The benchmark uses four convolutional network architectures from existing work~\citep{chiu2025sdp}:
    a ConvSmall trained on MNIST, and a ConvSmall, ConvDeep, and ConvLarge trained on CIFAR-10.
    We use motion-blur structured perturbations from VeriDou~\citep{duong2026verifying2} instead of per-pixel $\ell_\infty$ perturbations, for three reasons.
    First, standard local robustness perturbs each pixel independently inside an $\ell_\infty$ box,
    but real image corruptions such as motion blur, camera shake, and filtering are \emph{structured}:
    each output pixel is a weighted combination of its neighbors, so pixels move together~\citep{duong2026verifying,duong2026verifying2}.
    Second, an $\ell_\infty$ box cannot express this coupling exactly.
    Enclosing a structured perturbation in a pixel box admits images the transformation never produces,
    which leads to spurious counterexamples~\citep{duong2026verifying}.
    Third, structured perturbations expose weaknesses that $\ell_\infty$ noise misses:
    networks that appear highly robust under independent pixel noise are falsified on up to 99\% of instances once convolutional perturbations are added~\citep{duong2026verifying2}.

    VeriS~\citep{duong2026verifying} and VeriDou~\citep{duong2026verifying2} make such perturbations verifiable by standard tools.
    They prepend an affine layer that maps the perturbation parameters to the perturbed image, which turns structural robustness into ordinary local robustness over those parameters.
    We follow VeriDou's universal convolutional parameterization.
    Each entry of a $5\times5$ blur kernel is an independent variable, bounded in $[0, 1/5]$ on the positions covered by blur lines in the angle range and fixed to $0$ elsewhere.
    This covers a continuous range of angles and is strictly more expressive than restricted formulations that interpolate between a few fixed kernels.
    As a result, the \texttt{linf} property is a box over kernel entries rather than over pixels,
    and the halfspace and $\ell_2$ constraints restrict which blur kernels are allowed.
    We sweep the kernel angle over $[0^\circ, \theta_\text{max}]$ for six strengths $\theta_\text{max} \in \{15^\circ, 30^\circ, 60^\circ, 90^\circ, 120^\circ, 150^\circ\}$.

    For each (model, perturbation strength) pair, we evaluate three property types:
    the base $\ell_\infty$ box alone (\texttt{linf}),
    the box intersected with a halfspace constraint (\texttt{linf+hs}),
    and the box intersected with an $\ell_2$ ball constraint (\texttt{linf+l2}).
    For each base perturbation, the halfspace is $\sum x_i \le \sum c_i$ (a hyperplane through the box center)
    and the $\ell_2$ ball is centered at the upper corner $c + r$ with radius $\|r\|_2$,
    so all three property types share the same perturbation region and are directly comparable.
    We sample $3$ correctly-classified test images per combination, each yielding $9$ output disjunctions (one per non-true class),
    giving $4 \times 6 \times 3 \times 3 \times 9 = 1{,}944$ instances in total.

\end{document}